\documentclass[11pt,leqno]{amsart}
\usepackage{amsmath}
\usepackage{amssymb}
\usepackage{enumerate}
\usepackage{hyperref}

\usepackage{picins,graphicx}
\usepackage{tikz}
\usetikzlibrary{shapes.geometric, arrows}
\usepackage{pgfgantt}

\newtheorem{theorem}{Theorem}
\newtheorem{lemma}{Lemma}

\newtheorem{proposition}{Proposition}

\newcounter{claim_nb}[theorem]
\def\ba{\begin{array}}
\def\ea{\end{array}}
\def\beq{\begin{equation}}
\def\eeq{\end{equation}}
\def\bea{\begin{eqnarray}}
\def\eea{\end{eqnarray}}
\def\beann{\begin{eqnarray*}}
\def\eeann{\end{eqnarray*}}

\def\R{\mathbb{R}}

\def\cN{\mathcal{N}}
\def\cP{\mathcal{P}}

\begin{document}

\title[Coffman--Sethi conjecture]{Approximation ratio of LD algorithm\\ for multi-processor scheduling\\
and the Coffman--Sethi conjecture}

\author{Peruvemba Sundaram Ravi, Levent Tun\c{c}el}
\thanks{Peruvemba Sundaram Ravi: 
Operations and Decision Sciences,
School of Business and Economics, Wilfrid Laurier University,
Waterloo, Ontario N2L 3C5, Canada
(e-mail: pravi@wlu.ca)\\
Levent Tun\c{c}el: Department of Combinatorics and Optimization, Faculty
of Mathematics, University of Waterloo, Waterloo, Ontario N2L 3G1,
Canada (e-mail: ltuncel@uwaterloo.ca)}

\date{May 6, 2015}

\begin{abstract}
Coffman and Sethi proposed a heuristic algorithm, called LD, for multi-processor scheduling, to minimize makespan over flowtime-optimal schedules.
LD algorithm is a natural extension of a very well-known list scheduling algorithm, Longest Processing Time (LPT) list scheduling, to our bicriteria scheduling problem.
Moreover, in 1976, Coffman and Sethi conjectured that LD algorithm has precisely the following
worst-case performance bound: $\frac{5}{4}-\frac{3}{4(4m-1)}$, where $m$ is the number of machines. In this paper, utilizing some recent work by the authors
and Huang, from 2013, which exposed some very strong combinatorial properties of various presumed minimal counterexamples to the conjecture,
we provide a proof of this conjecture.
The problem and the LD algorithm have connections to other fundamental problems (such as the assembly line-balancing problem) and to other algorithms.
\end{abstract}


\maketitle

\section{Introduction}
\label{sec:intro}

The most fundamental machine environment in multiprocessor scheduling problems is a \emph{parallel identical machine model}.
In this basic set-up, we have $m$ parallel identical machines that are simultaneously available at time zero, and $n$ independent jobs,
all simultaneously available at time zero,
indexed by $1,2, \ldots, n$ with given processing times $p_1, p_2, \ldots, p_n$.  No pre-emption is allowed, and the machines are assumed
to be completely reliable.  For a scheduling problem environment described above, with data $m, p_1, p_2, \ldots, p_n$, there are two
performance criteria that immediately come to mind:
\begin{itemize}
\item
minimize the completion time of the last job (i.e., \emph{makespan}),
\item
minimize the total (or equivalently the average) time that the jobs spend in the system
(i.e., total or average \emph{flowtime}).
\end{itemize}
Given a feasible schedule, let $C_j$ denote the completion time of job $j$ in that schedule.  By denoting
$C_{\max}:= \max_{j \in \{1,2, \ldots, n\}} \left\{ C_j \right\}$, our two criteria are:
\begin{itemize}
\item
minimize $C_{\max}$,
\item
minimize $\sum_{j=1}^n C_j$.
\end{itemize}
Both of these objective functions are easily justifiable.  Indeed, the minimization of makespan may ensure optimal utilization
of resources (i.e., machines) as well as ensuring the earliest possible start times for other tasks that require the completion of
all the jobs $1, 2, \ldots, n$ to be started.  Minimization of total flow time $F:= \sum_{j=1}^n C_j$, minimizes
the amount of time the jobs spend in the system (in our setting this is $C_j$ for each job $j$).  Thus, minimizing $F$
equivalently minimizes, in many applications, work-in-process inventory.  A feasible schedule is
called \emph{flowtime-optimal} if it minimizes $F$. In this paper, we consider
the bicriteria optimization problem of minimizing makespan among all flowtime-optimal schedules.  In scheduling theory
notation, let $F^*$ denote the optimal objective function value of $Pm \,\, / \,\,\,\, / \,\, \sum C_j$.  Then, our
bicriteria optimization problem is: $Pm \,\, / \,\,\,\, / \,\, C_{\max} ; \sum C_j=F^*$, which we call \emph{Flowtime-Makespan (FM)}
problem.

There are two single objective function scheduling problems that make up our bicriteria optimization problem:
$Pm \,\, / \,\,\,\, / \,\, C_{\max}$, and $Pm \,\, / \,\,\,\, / \,\, \sum C_j$.
The second problem is as easy as sorting and, as a result, admits algorithms with $O(n \log(n))$ complexity.  Moreover,
we have a complete characterization of all optimal solutions of $Pm \,\, / \,\,\,\, / \,\, \sum C_j$.
Conway, Maxwell and Miller \cite{ConwayMaxwellMiller1967}, in their seminal book, develop the notion of \emph{rank}
for the FM problem. To simplify presentation, we may assume that $m$ divides $n$ (otherwise, we can add
$(m\left\lceil n/m\right\rceil - n)$ dummy jobs with zero processing times; this modification does not
affect the conclusions of this paper).  Then, in such an instance, there are $k:=n/m$ \emph{ranks}.  We may further
assume that the jobs are indexed in nonincreasing order of processing times
(i.e., with job $1$ having the largest processing time), the set of jobs belonging to \emph{rank} $r$ are the following:
${(r-1)m + 1, (r-1)m + 2,\cdots, (r - 1)m + m}$.
Then a feasible schedule is \emph{flowtime-optimal} if jobs are assigned in decreasing order of ranks,
with the jobs in rank $1$ being assigned last. Since within each rank the assignment of jobs to machines
can be arbitrary, it immediately follows that there are
at least $\left({m!}\right)^{\lfloor n/m \rfloor}$ flowtime-optimal schedules.  From a mathematical viewpoint,
it makes sense to consider a secondary criterion to choose a ``best'' flowtime-optimal schedule among these
huge number of schedules everyone of which minimizes total flowtime.  Also, this is reasonable from a practical viewpoint.

The first problem, $Pm \,\, / \,\,\,\, / \,\, C_{\max}$, is $\cN\cP$-hard even for $m=2$ (trivial reduction from PARTITION).
Graham's ground-breaking work on the subject in 1960's tackled the problem\\
$Pm \,\, / \,\,\,\, / \,\, C_{\max}$.
This work was ground-breaking not only in approximation
algorithms for scheduling, but in approximation algorithms in general. Graham first proved:
\begin{theorem}(Graham \cite{Graham1966})
\emph{List Scheduling} algorithm
has worst-case approximation ratio of $\left(2-\frac{1}{m}\right)$. Moreover, this bound is achievable for every $m \geq 2$.
\end{theorem}
Then, Graham analyzed the List Scheduling algorithm when the list is given in LPT order and provided a very elegant
proof of the following result:
\begin{theorem}(Graham \cite{Graham1969})
LPT-List Scheduling algorithm
has worst-case approximation ratio of $\left(\frac{4}{3}-\frac{1}{3m}\right)$. Moreover, this bound is achievable for every $m \geq 2$.
\end{theorem}

Just like the scheduling problem $Pm \,\, / \,\,\,\, / \,\, C_{\max}$, the FM problem is also $\cN\cP$-hard (a result of
Bruno, Coffman and Sethi \cite{BCS1974}).
In 1976, Coffman and Sethi \cite{CoffmanSethi1976a} proposed some approximation algorithms for the FM problem. Among these algorithms, in
the LD algorithm (which is an extension of the LPT list scheduling for the FM problem), ranks are assigned in increasing order,
starting with rank one, which is the rank containing the set of $m$ jobs with the largest processing times. Jobs within
the same rank are assigned largest-first onto distinct machines as the machines become
available after completing the jobs in the previous ranks. After all the jobs are assigned, the schedule
is reversed and all jobs in the last rank (that is, rank $k$) must be set to the same starting time of zero.
Coffman and Sethi conjectured the following worst-case bound for the LD algorithm.

\begin{center}
\begin{tabular}{|c|}\hline
\\
\parbox{5in}{\underline{Coffman-Sethi conjecture} \cite{CoffmanSethi1976a}:\\
the LD algorithm has a makespan ratio with a worst-case bound equal to
$$\frac{5m-2}{4m-1}=\frac{5}{4} -\frac{3}{4(4m-1)}.
\,\,\,\,\,\,\,\,\,\,\,\,\,\,\,\,\,\,\,\,\,\,\,\,\,\,\,\,\,\,\,\,\,\,\,
\,\,\,\,\,\,\,\,\,\,\,\,\,\,\,\,\,\,\,\,\,\,\,\,\,\,\,\,\,\,\,\,\,\,\,\,\,\,\,\,\,\,\,\,\,\,\,$$}\\
\\ \hline
\end{tabular}
\end{center}

\vspace{0.2cm}

The authors and Huang \cite{RTH2013} constructed the following family of
instances, proving that the above conjectured ratio cannot be improved
for any $m \geq 2$.
For every integer $m \geq 2$, let $n:=3m$ and define the processing times as:
\[
p_j := \left\{\begin{array}{cl} 0, & \mbox{ for } j \in \{1,2,
\ldots,m-1\};\\
m, & \mbox{ for } j =m;\\
(j-1), & \mbox{ for } j \in \{m+1,m+2, \ldots, 2m\};\\
(j-2), & \mbox{ for } j \in \{2m+1, 2m+2, \ldots, 3m\}.
\end{array}
\right.
\]
It is easy to verify that the ratio of the objective value of an LD
schedule to the optimal objective value is exactly $\frac{5m-2}{4m-1}$,
for every integer $m \geq 2$.  In Figure \ref{fig:1}, we present an LD schedule
for the three-machine instance of this family of bad instances.  Note that each of
the machines 1 and 2 have a job (either job 1 or job 2) with processing time zero,
started and finished at time zero (not shown on the Gantt chart).  Completion
times are 10, 10, and 13 on the machines 1, 2, and 3 respectively.
\begin{figure}[h]
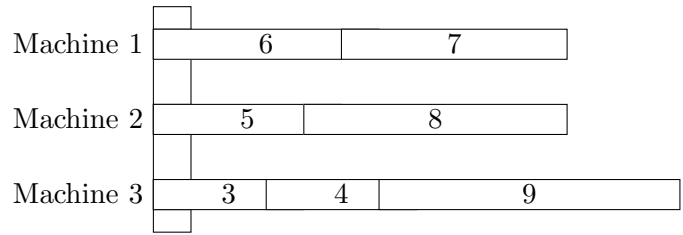

\begin{center}
\begin{ganttchart}{1}{16}
\ganttbar{Machine 1}{1}{1}
\ganttbar[inline]{6}{1}{6}
\ganttbar[inline]{7}{6}{11}\\
\ganttbar{Machine 2}{1}{1}
\ganttbar[inline]{5}{1}{5}
\ganttbar[inline]{8}{5}{11}\\
\ganttbar{Machine 3}{1}{1}
\ganttbar[inline]{3}{1}{4}
\ganttbar[inline]{4}{4}{7}
\ganttbar[inline]{9}{7}{14}
\end{ganttchart}
\caption{LD schedule for $m=3$ \label{fig:1}}
\end{center}
\end{figure}
Figure \ref{fig:2} presents the Gantt chart for an optimal schedule, with optimal
makespan 11, for the
same three-machine instance.  Again, machines 1 and 2 have jobs with zero processing times.
Further note that this schedule is ``rectangular,'' providing an obvious certificate
of the optimality of the underlying makespan.  This term ``rectangular schedule,''
is rigorously defined in the next section.
\begin{figure}[h]
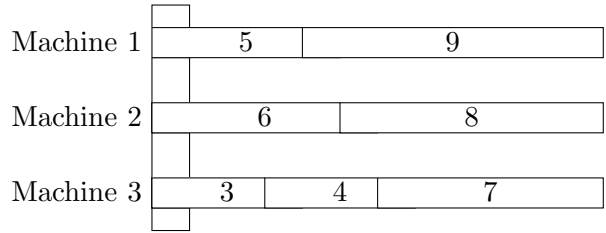

\begin{center}
\begin{ganttchart}{1}{16}
\ganttbar{Machine 1}{1}{1}
\ganttbar[inline]{5}{1}{5}
\ganttbar[inline]{9}{5}{12}\\
\ganttbar{Machine 2}{1}{1}
\ganttbar[inline]{6}{1}{6}
\ganttbar[inline]{8}{6}{12}\\
\ganttbar{Machine 3}{1}{1}
\ganttbar[inline]{3}{1}{4}
\ganttbar[inline]{4}{4}{7}
\ganttbar[inline]{7}{7}{12}
\end{ganttchart}
\caption{Optimal schedule for $m=3$ \label{fig:2}}
\end{center}
\end{figure}

The Coffman-Sethi conjecture has remained open for nearly four decades. The first major
progress was made by the authors and Huang in \cite{RTH2013}.
In the next section, utilizing our recent work with Huang \cite{RTH2013},
we provide a proof of this conjecture.  Our approach in obtaining a proof of Coffman and Sethi's
conjectured bound is to identify properties of a hypothesized minimal counterexample to the conjecture. To define
our notion of minimality, we use a five-level hierarchical ordering of various attributes of the FM problem instances.
Using integrality of the data and
some more sophisticated properties of the minimal counterexamples, we construct some ``smaller'' problem instances by subtracting nonnegative integers from
the number of machines, the number of jobs, or the integer-valued processing times. This exposes some very strong combinatorial structures of presumed
minimal counterexamples.  These special combinatorial structures lead us to very strong constraints that capture the fact that all
problem instances that are ``smaller'' than the minimal counterexample at hand, must satisfy the conjecture.
This eventually leads to a contradiction for problem instances with four or more ranks.  Given that a minimal counterexample cannot exist for one,
two or three ranks, or with two or three machines as was previously proved by the authors and Huang \cite{RTH2013},
these results prove that the conjecture holds for all instances of problem FM.

Note that there are more sophisticated algorithms with better approximation guarantees for the kind of scheduling problems (and in general
combinatorial optimization problems) that we are considering in this paper.  The bin-packing problem provides another combinatorial
optimization context which is helpful in attacking makespan minimization problems in multi-processor scheduling.  First Fit Decreasing
(FFD) algorithm \cite{Johnson1973} (also see \cite{JohnsonDemersUllmanGareyGraham1974}) admits good approximation bounds.  Various alternatives
to FFD were proposed \cite{GareyJohnson1981,CoffmanGareyJohnson1983}.  Moreover, Coffman, Garey and Johnson \cite{CoffmanGareyJohnson1978} designed the MultiFit algorithm
(an approximation algorithm for minimizing makespan) based on FFD and proved a worst-case performance bound.  Recently the exact performance bound
of MultiFit was established as 24/19 by Hwang and Lim \cite{HwangLim2014}. de la Vega and Lueker \cite{delaVegaLueker1981} propose a linear-time
approximation scheme for the bin packing problem.  For some other work on MultiFit
algorithm and related topics, see \cite{ChangHwang1999,Dosa2000,Dosa2001,Friesen1984,Yue1990}.  Hochbaum and Shmoys \cite{HochbaumShmoys1987} propose a PTAS (polynomial-time approximation scheme)
for the makespan minimization problem on the parallel identical machine model.  Eck and Pinedo \cite{EckPinedo1993} propose a new algorithm LPT*
for the FM problem and for the two machine case, prove the worst-case approximation ratio of 28/27.  A slight generalization of the FM problem can be
formulated as the problem of permuting the elements within the columns of a $m$-by-$n$ matrix with nonnegative entries, so as to minimize its
maximum row sum.  This problem, which models the assembly line balancing problem, was studied by Coffman and Yannakakis \cite{CoffmanYannakakis1980} as well as Hsu \cite{Hsu1984}.
Therefore, FM problem is a fundamental problem in multi-processor scheduling, with close ties to other problems in combinatorial optimization.  The LD algorithm
is also closely related to some fundamental heuristics for such problems.

In the next section, we provide a detailed, rigorous definition of the LD algorithm.  In Section \ref{sec:3}, we present a proof of the Coffman--Sethi conjecture.

\section{LD Algorithm, definitions and some properties}
\label{sec:2}

Let us index the jobs in nonincreasing order of processing times.  The set of jobs belonging to rank $r$ is:
\[
\{(r-1)m + 1, (r-1)m + 2,\ldots, (r - 1)m + m\}.
\]

A feasible schedule in which all rank $(r+1)$ jobs are started before all rank $r$ jobs (where $r \in \{ 1,2,\cdots, (n/m) - 1\}$) is said to satisfy the \emph{rank restriction} or \emph{rank constraint}.
A feasible schedule without idle time and satisfying the rank constraint, and in which all rank $n/m$ jobs start
at time zero, is a \emph{flowtime-optimal schedule}.

In every rank $r$, we identify the largest and smallest processing times and denote them by $\lambda_{r}$ and $\mu_{r}$. Therefore, we have
\bea \lambda_1 \geq \mu_1 \geq \lambda_2 \geq \mu_2 \geq \cdots \geq \lambda_{k-1} \geq \mu_{k-1} \geq \lambda_k \geq \mu_k \geq 0, \eea
where, $k:= \lceil \frac{n}{m} \rceil$.
The \emph{profile} of a schedule after rank $r$ is defined as the sorted set of completion times (in nonincreasing order)
on $m$ machines after rank $r$.  When all the jobs in first $r$ ranks (out of a total of $k$ ranks) have been assigned to machines, and the jobs in the
remaining ranks have not yet been assigned to machines, the profile after rank $r$ is called the \emph{current profile}.
We denote the current profile by $a(r) \in \R^m: a_1(r) \geq a_2(r) \geq \cdots \geq a_m(r)$. I.e., 
$a_{i}(\ell)$ denotes the $i$th largest completion time in the schedule, after rank $\ell$.

The LD algorithm of Coffman and Sethi \cite{CoffmanSethi1976a} schedules the ranks in the following order:
$1, 2 , \ldots, k-1, k.$
Denote the current profile by $a \in \R^m$: $a_1 \geq a_2 \geq \cdots \geq a_m$.
Then, schedule the jobs in the next rank so that the job with the largest processing time is matched with $a_m$
(the smallest part of the current profile), second largest processing time is matched with $a_{m-1}$
(the second smallest part of the current profile), etc., and the smallest processing time is matched with $a_1$
(the largest part of the current profile). After all the
jobs are scheduled, the schedule is reversed and left-justified (i.e., we start the first job on each machine at time zero).

As we mentioned in the introduction, without loss of generality,
we may assume the property
\[
\textup{(Property.1)} \,\,\,\,\,\,\,\,\,\, n=mk,
\]
while allowing some jobs to have a zero processing time.
Given an instance of the FM problem, we denote by $t_{LD}$ the makespan of the LD schedule(s).  We use $t^*$ to denote the makespan of the optimal schedule(s).
The following lemma is stated without proof. A proof is provided in \cite{Ravi2010}.

\begin{lemma} \label{LD_TYPE_I}
(See \cite{RTH2013,Ravi2010})
For the FM problem and the LD algorithm, the following must hold: If there exists a counterexample to a conjectured $t_{LD}/t^*$ ratio, then there exists a counterexample with integer
processing times.\end{lemma}

We will, therefore, restrict our attention to problem instances with integer data in this paper. Note that due to integrality of the
processing times, the smallest nonzero processing time is bounded below by one.
We define the ordered set of processing times $P$ for a scheduling problem instance with $m$ machines and $k$ ranks to consist of elements equal to the processing times of these $mk$ jobs arranged in nonincreasing order. We
use $P(j)$ to refer to the $j^{th}$ entry of $P$.

In our characterizations of minimal counterexamples to the Coffman--Sethi conjecture, minimality is defined based on the attributes: the number of machines, the number of ranks,
and the ordered set of processing times with respect to a five-level, hierarchical grading.  See \cite{RTH2013} for the definition minimality
for the classes of counterexamples we consider.

Throughout this paper, we will isolate many useful properties of minimal counterexamples of various types.
Note that, as it was established in \cite{RTH2013}, we may assume that in a minimal counterexample, the following property holds:
\[
\textup{(Property.2)} \,\,\,\,\,\,\,\,\, \mu_r=\lambda_{r+1}, \forall r \in \{1,2, \ldots, k-1\} \mbox{ and } \mu_k=0.
\]

We continue with some more definitions.
\begin{itemize}
 \item
 A problem instance of \emph{Type I} is a problem instance that has integer processing times.
 \end{itemize}

A counterexample of Type I is an FM problem instance of that type which violates the Coffman--Sethi conjecture. A \emph{minimal counterexample} of Type I is a counterexample of that type for which
there does not exist a smaller counterexample (based on the notion of minimality defined in \cite{RTH2013}) of the Type I.

A \emph{rectangular schedule} is a feasible schedule in which all machines are busy between time zero and the makespan.  Note
that every rectangular schedule minimizes the makespan, since its objective function value matches an obvious lower bound of $\sum_{j=1}^n p_j /m$
on the makespan of every feasible schedule.

 \begin{itemize}
 \item
 A problem instance of \emph{Type IR} is one with integer processing times and a rectangular optimal schedule.
 \end{itemize}

 A counterexample of Type IR is defined analogously.

\section{A proof of the Coffman--Sethi conjecture}
\label{sec:3}

 We state without proofs the following two lemmas from previous work.

\begin{lemma} \label{NONDECREASING_PROFILE} (Ravi, Tun\c{c}el and Huang \cite{RTH2013})
An increase in one or more processing times of jobs in rank $r$ for $r \in\{1,2, \ldots,k -1\}$ (with no change in the remaining processing times, and subject to the rank constraint) does
not result in a reduction in any element of the profile $a(\ell)$ of an LD schedule after rank $\ell \in \{r, r+1, \ldots, k\}$.
\end{lemma}

\begin{lemma}
\label{cor:IR}(Ravi, Tun\c{c}el and Huang \cite{RTH2013})
If the Coffman--Sethi conjecture is false, then there exists a minimal counterexample to the conjecture of Type IR.
\end{lemma}

Next, we describe two useful ways (procedures REDUCE(P1,$r$) and $\Box$-REDUCE(P1,$r$)) of generating ``smaller'' FM instances from a given FM instance.
Let P1 denote an FM problem instance.

REDUCE(P1,$r$): Construct P2 from P1 by subtracting one time unit from the processing time of every job in rank $r-1$ and subtracting one time unit from the processing time of every job in rank $r$ that has a processing time of $\lambda_{r}$. Leave the remaining processing times unchanged.

Figure \ref{fig:3} presents an LD schedule $S$ for $m:=3$, $P1:=[9,8,7,7,6,5,5,2,1]$. Completion
times are 15, 16, and 19 on the machines 1, 2, and 3 respectively.

\begin{figure}[h]
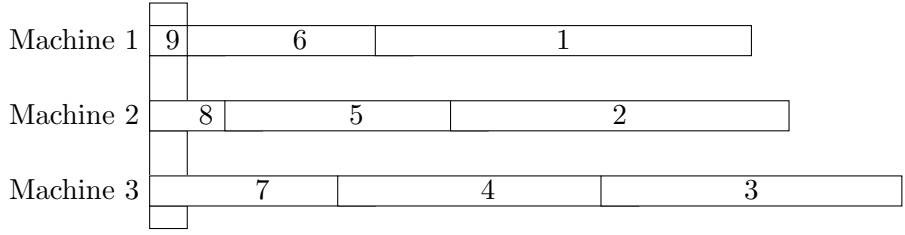

\begin{center}
\begin{ganttchart}{1}{24}
\ganttbar{Machine 1}{1}{1}
\ganttbar[inline]{9\,\,\,\,\,\,}{1}{2}
\ganttbar[inline]{6}{2}{7}
\ganttbar[inline]{1}{7}{16}\\
\ganttbar{Machine 2}{1}{1}
\ganttbar[inline]{8}{1}{3}
\ganttbar[inline]{5}{3}{9}
\ganttbar[inline]{2}{9}{17}\\
\ganttbar{Machine 3}{1}{1}
\ganttbar[inline]{7}{1}{6}
\ganttbar[inline]{4}{6}{13}
\ganttbar[inline]{3}{13}{20}
\end{ganttchart}
\caption{An LD schedule $S$ for the instance given by P1 \label{fig:3}}
\end{center}
\end{figure}

 Figure \ref{fig:4} presents the result of the application of REDUCE(P1,2) to the original LD schedule $S$,
yielding REDUCE(P1,2) = $[8,7,6,6,6,5,5,2,1]$, and the completion times become 14, 15, and 17 on the machines 1, 2, and 3 respectively.

\begin{figure}[h]
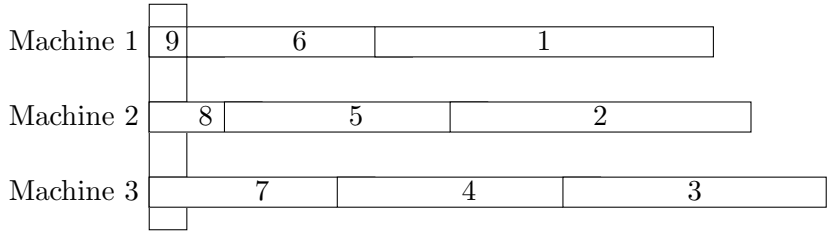

\begin{center}
\begin{ganttchart}{1}{24}
\ganttbar{Machine 1}{1}{1}
\ganttbar[inline]{9\,\,\,\,\,\,}{1}{2}
\ganttbar[inline]{6}{2}{7}
\ganttbar[inline]{1}{7}{15}\\
\ganttbar{Machine 2}{1}{1}
\ganttbar[inline]{8}{1}{3}
\ganttbar[inline]{5}{3}{9}
\ganttbar[inline]{2}{9}{16}\\
\ganttbar{Machine 3}{1}{1}
\ganttbar[inline]{7}{1}{6}
\ganttbar[inline]{4}{6}{12}
\ganttbar[inline]{3}{12}{18}
\end{ganttchart}
\caption{LD schedule $S_1$ for the instance REDUCE(P1,2) \label{fig:4}}
\end{center}
\end{figure}

$\Box$-REDUCE(P1,$r$): Construct P2 from P1 by applying the procedure REDUCE(P1,$r$) to P1. Construct P2R from P2 as follows. For every job in rank $1$ of the optimal schedule for P2 that is processed on a machine with a completion time after rank $k$ that is less than the makespan, increase the processing time so that the completion time after rank $k$ becomes equal to the makespan.

Note that every instance generated by $\Box$-REDUCE(P1,$\cdot$) has, by construction, a rectangular
optimal schedule.
Figure \ref{fig:5} presents an optimal schedule for the instance given by P2 := $[8,7,6,6,6,5,5,2,1]$.  In this
optimal schedule, the completion times are 15, 15, and 16 on the machines 1, 2, and 3 respectively.

\begin{figure}[h]
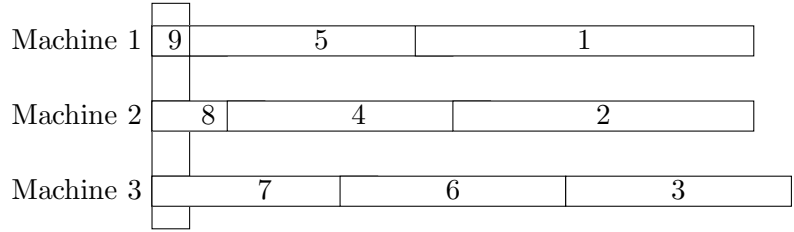

\begin{center}
\begin{ganttchart}{1}{24}
\ganttbar{Machine 1}{1}{1}
\ganttbar[inline]{9\,\,\,\,\,\,}{1}{2}
\ganttbar[inline]{5}{2}{8}
\ganttbar[inline]{1}{8}{16}\\
\ganttbar{Machine 2}{1}{1}
\ganttbar[inline]{8}{1}{3}
\ganttbar[inline]{4}{3}{9}
\ganttbar[inline]{2}{9}{16}\\
\ganttbar{Machine 3}{1}{1}
\ganttbar[inline]{7}{1}{6}
\ganttbar[inline]{6}{6}{11}
\ganttbar[inline]{3}{12}{17}
\end{ganttchart}
\caption{An optimal schedule for the instance given by P2 \label{fig:5}}
\end{center}
\end{figure}

Figure \ref{fig:6} presents an optimal schedule for the instance $\Box$-REDUCE(P1,2).
Here, $\Box$-REDUCE(P1,2) = $[9,8,6,6,6,5,5,2,1]$ whose optimal schedules yield
a makespan of 16 on every machine.

\begin{figure}[h]
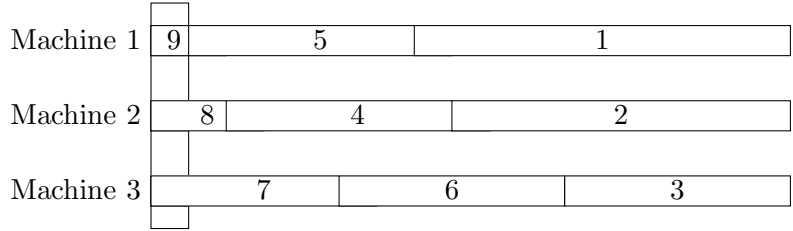

\begin{center}
\begin{ganttchart}{1}{24}
\ganttbar{Machine 1}{1}{1}
\ganttbar[inline]{9\,\,\,\,\,\,}{1}{2}
\ganttbar[inline]{5}{2}{8}
\ganttbar[inline]{1}{8}{17}\\
\ganttbar{Machine 2}{1}{1}
\ganttbar[inline]{8}{1}{3}
\ganttbar[inline]{4}{3}{9}
\ganttbar[inline]{2}{9}{17}\\
\ganttbar{Machine 3}{1}{1}
\ganttbar[inline]{7}{1}{6}
\ganttbar[inline]{6}{6}{11}
\ganttbar[inline]{3}{12}{17}
\end{ganttchart}
\caption{An optimal schedule for the instance $\Box$-REDUCE(P1,2) \label{fig:6}}
\end{center}
\end{figure}

\begin{proposition} \label{LD_makespan_of_P2}
\begin{enumerate}
 \item Let P1 be a minimal counterexample of Type I and
 \textup{P2}:= REDUCE(\textup{P1},$r$).  Then, $t_{LD}(\textup{P2}) \leq t_{LD}(\textup{P1})-2.$
\item Let \textup{P1} be a minimal counterexample of Type IR and
\textup{P2R}$:= \Box$-REDUCE(\textup{P1},$r$).\\
Then, $t_{LD}(\textup{P2R}) \leq t_{LD}(\textup{P1})-2.$
 \end{enumerate}
\end{proposition}
\begin{proof} We prove the second assertion only (the proof of the first assertion is similar).  For the purpose of a contradiction,
assume that the proposition is false. Therefore, the LD makespan of P2R is at least $t_{LD} - 1$,
where $t_{LD}$ denotes the makespan of the LD schedule for P1. It is evident that P2R has an optimal makespan that is at most $t^{*} - 1$, where $t^{*}$ denotes the optimal makespan of P1. Therefore P2R is a problem instance of Type IR with a worse approximation ratio than P1 and the same number of ranks and machines and jobs with nonzero processing times as P1. Thus, we have a contradiction and therefore the proposition must be true.\end{proof}

We define a minimal counterexample of Type IR1 as follows. If the Coffman--Sethi conjecture is false, a minimal counterexample of Type IR1 is a minimal counterexample of Type IR that has an LD schedule with the property that
every machine with a completion time after rank $k$ equal to the makespan has a job with processing time equal to $\lambda_k$ in rank $k$.  (Recall, $k$ denotes the number of ranks.)

\begin{lemma} \label{IR1_EXISTS}
(Ravi, Tun\c{c}el and Huang \cite{RTH2013})
If the Coffman--Sethi conjecture is false, then there exists a minimal counterexample to the conjecture of Type IR1, and every minimal counterexample of Type IR is a minimal counterexample of Type
IR1.\end{lemma}

Below, we give a proof of this lemma utilizing the procedure REDUCE and $\Box$-REDUCE to
give a glimpse of how to establish properties of Type IR1 and other related minimal counterexamples.

\begin{proof} Suppose the Coffman--Sethi conjecture is false.
Then, by Lemma \ref{cor:IR}, a counterexample of Type IR
exists. Suppose, for the purpose of a contradiction that there exists a minimal
counterexample P1 of Type IR that is not Type IR1.
Now construct a problem instance P2R of Type IR by applying the procedure $\Box$-REDUCE(P1,$k$) to P1. The first step of the procedure is REDUCE(P1,$k$) which produces a problem instance P2. This step leaves the assignment of jobs in each rank to machines in the LD schedule unchanged. Therefore
 P2 has an LD makespan of $t_{LD} - 1$, where $t_{LD}$ denotes the LD makespan of P1. By Lemma \ref{NONDECREASING_PROFILE}, the LD makespan of P2R cannot be less than the LD makespan of P2. Therefore it is greater than or equal to $t_{LD} - 1$. However, by Proposition \ref{LD_makespan_of_P2}, the LD makespan of P2R is less than or equal to $t_{LD} - 2$. Therefore, we have a contradiction and the lemma must hold.\end{proof}

Our proof of the main result of this paper uses the concept of Type I2 instances
as established in \cite{RTH2013}.  The proof that ``if Coffman--Sethi conjecture is false then a counterexample
of Type I2 exists'' uses procedures like REDUCE and $\Box$-REDUCE and the properties
of Type IR1 counterexamples.  Next, we define Type I2 instances.

If the Coffman--Sethi conjecture is false, a counterexample to the conjecture of \emph{Type I2} is a counterexample of Type I that has an LD schedule with the
following properties:
\begin{enumerate}[(i)]
\item
It has only one machine $i'$ with a completion time after rank $k$ equal to the makespan.
\item
Machine $i'$ has a processing time equal to $\lambda_r$ in rank $r$ for every $r \in \{2,3, \ldots, k\}$.
\end{enumerate}

\begin{lemma} \label{lem:I2properties} (Ravi, Tun\c{c}el and Huang \cite{RTH2013})
If the Coffman--Sethi conjecture is false, then there exists a minimal counterexample to the conjecture of Type I2.  Moreover, in a minimal counterexample of Type I2, the following properties hold:
\begin{itemize}
\item
the sole machine $i'$ with a completion time after rank $k$ equal to the makespan
in the LD schedule has a processing time equal to $\mu_1$ in rank $1$;
\item
there exists at least one machine $i''$ with $i'' \neq i'$, such that
the completion time after rank $(k - 1)$ on machine $i''$ is greater
than or equal to the completion time after rank $(k - 1)$ on machine $i'$;
\item
the smallest completion time after rank $k$ on any machine is
at least
\[t_{LD} - \displaystyle{\max_{r \in \{2, 3, \ldots ,k\}} \left\{ \lambda_r - \mu_r \right\}}.
\]
\end{itemize}
 \end{lemma}

\begin{proof}
See Lemmas 6, 7, 8, 9, and 10 of \cite{RTH2013}.
\end{proof}

Note that Lemma \ref{lem:I2properties} exposes many, combinatorially very strong properties of a minimal
counterexample of Type I2.  These properties allow us to have access to very strong inequalities
on the optimal makespan in terms of the makespan of an LD schedule for such minimal counterexamples.

\begin{theorem} \label{thm:1} (Ravi, Tun\c{c}el and Huang \cite{RTH2013})
The Coffman--Sethi conjecture holds for all instances with either property
given below:
\begin{itemize}
\item[(i)]
$m \leq 3$ (FM instances with at most three machines),
\item[(ii)]
$k \leq 3$
(FM instances with at most three ranks, i.e., for all machine-job pairs $(m,n)$
satisfying $n \leq 3m$).
\end{itemize}
\end{theorem}

\begin{proof}
See, respectively, Theorems 2 and 3 of \cite{RTH2013}.
\end{proof}

Next, we prove that the conjecture holds for all the remaining cases.

\begin{theorem} \label{thm:2}
The Coffman--Sethi conjecture holds for all instances with either property
given below:
\begin{itemize}
\item[(i)]
$m \geq 4$ (FM instances with at least four machines),
\item[(ii)]
$k \geq 4$
(FM instances with at least four ranks, i.e., for all machine-job pairs $(m,n)$
satisfying $n \geq 3m+1$).
\end{itemize}
\end{theorem}

\begin{proof}
Suppose the above claim is false.  Then there exists a minimal counterexample of Type I2 to the
Coffman--Sethi conjecture.  Since the conjecture holds for all instances with $m\leq 3$
as well as for all instances with $k \leq 3$ (by Theorem \ref{thm:1}), there must exist a minimal counterexample of Type I2
to the claim with $k$ and $m$ both at least equal to four.  Let $t$ denote the makespan for an LD schedule of the minimal counterexample of Type I2 with $k$ ranks.
Then, by Lemma \ref{lem:I2properties}, we have
\[
mt^* \geq t + (t - \lambda_k) + (m - 2)\left(t - \max_{r \in \{2, \ldots, k\}} \{\lambda_r - \mu_r\}\right).
\]
The last relation is equivalent to
\bea
\label{eqn:a}
t^* & \geq & t - \frac{\lambda_k}{m} -\left(1-\frac{2}{m}\right) \max_{r \in \{2, \ldots, k\}} \{\lambda_r - \mu_r\}.
\eea
Since we are working with a counterexample,
\bea
\label{eqn:badratio}
t & > & \left(\frac{5m-2}{4m-1}\right) t^*.
\eea
Inequalities \eqref{eqn:a} and \eqref{eqn:badratio} imply,
\bea
\label{eqn:c}
\left(\frac{m-1}{4m-1}\right) t^* & < & \frac{\lambda_k}{m} + \left(1-\frac{2}{m}\right) \max_{r \in \{2, \ldots, k\}} \{\lambda_r - \mu_r\}.
\eea
Suppose that the maximum, $\max_{r \in \{2, \ldots, k\}} \{\lambda_r - \mu_r\}$, is attained by $r=k$.  Then, \eqref{eqn:c} implies
(since $\mu_k =0$): $t^* < \left(4-\frac{1}{m}\right) \lambda_k.$  However, $t^* \geq \lambda_k + \sum_{r=1}^{k-1}\mu_r > k \lambda_k.$
Since $k \geq 4$, we reach a contradiction.  Therefore, we may assume, there exists $s \in \{2,3, \ldots, k-1\}$ such that
$\max_{r \in \{2, \ldots, k\}} \{\lambda_r - \mu_r\} = \lambda_s -\mu_s.$

Using \eqref{eqn:a} and \eqref{eqn:badratio}, as well as the facts $\mu_s = \lambda_{s+1}$ and
$\lambda_k \leq \lambda_{s+1}$, we obtain
\bea
\label{eqn:d}
\left(\frac{m-1}{5m-2}\right) t & < & \frac{\lambda_{s+1}}{m} + \left(1-\frac{2}{m}\right) (\lambda_s - \lambda_{s+1}).
\eea
Since
\bea
\label{eqn:e}
t & = & 2 \lambda_2 + \sum_{r=3}^k \lambda_r,
\eea
substituting this for $t$ in \eqref{eqn:d}, we derive:
\bea
\label{eqn:f}
\lambda_2 & < & -\frac{1}{2} \sum_{r=3}^k \lambda_r + \frac{1}{2(m-1)} \left(5 -\frac{2}{m}\right) \lambda_{s+1}
+\left[\frac{m-2}{2(m-1)}\right] \left(5-\frac{2}{m}\right)(\lambda_s - \lambda_{s+1}).
\eea

Next, we consider a lower bound on $t^*$ based on $\lambda_s$:
\bea
\label{eqn:g}
t^* & \geq & \sum_{r=1}^{s-1} \mu_r + \lambda_{s} + \sum_{r=s+1}^{k} \mu_r
= \sum_{r=2}^{s-1} \lambda_r + 2\lambda_{s} + \sum_{r=s+2}^{k} \lambda_r.
\eea
Note that we are using the convention that an empty sum is zero.
Since we are working with a counterexample, we have $\frac{t}{t^*} > \frac{5m-2}{4m-1}$.  This,
together with the relations \eqref{eqn:g} and \eqref{eqn:e} imply
\bea
\label{eqn:ge}
\frac{2\lambda_{2} + \sum_{r=3}^{k} \lambda_r}{\sum_{r=2}^{s-1} \lambda_r + 2\lambda_{s} + \sum_{r=s+2}^{k} \lambda_r}
& > & \frac{5m-2}{4m-1}.
\eea
If $s \in \{3,4, \ldots,k-1\}$, then the last inequality is equivalent to
\bea
\label{eqn:h}
\lambda_2 & > & \frac{1}{3}\left(1-\frac{1}{m}\right) \sum_{r=3}^{s-1} \lambda_r + \left(2- \frac{1}{m}\right)\lambda_s
-\frac{1}{3}\left(4 -\frac{1}{m}\right) \lambda_{s+1} + \frac{1}{3}\left(1-\frac{1}{m}\right) \sum_{r=s+2}^{k} \lambda_r.
\eea
Finally, relations \eqref{eqn:f} and \eqref{eqn:h} imply
\beann
& & \left(\frac{5}{6} -\frac{1}{3m}\right) \left(\sum_{r=3}^{s-1} \lambda_r + \sum_{r=s+2}^{k} \lambda_r \right)
+\left[\frac{5}{2} -\frac{1}{m}-\frac{(5m-2)(m-2)}{2m(m-1)}\right] \lambda_s \\
& & +\left(\frac{5}{3}-\left[\frac{17m -8}{3m(m-1)}\right]\right) \lambda_{s+1}
< 0.
\eeann
For $m$ and $k$ at least four, sum of the first two terms on the left-hand-side is clearly positive.  The last term is nonnegative
for every $m \geq 4$.  Hence, we reached a contradiction.  Therefore, we may assume, $s=2$.

Let us go back to relation \eqref{eqn:c} and use $s=2$ and $\lambda_k \leq \lambda_4$ to obtain:
\bea
\label{eqn:d2}
\left(\frac{m-1}{5m-2}\right) t & < & \frac{\lambda_{4}}{m} + \left(1-\frac{2}{m}\right) (\lambda_2 - \lambda_{3}).
\eea
Substituting \eqref{eqn:e} into the above, we have
\bea
\label{eqn:f2}
\lambda_2 & > & \left( \frac{6m^2-13m+4}{3m^2-10m+4} \right) \lambda_{3}
+\left(\frac{m^2-6m+2}{3m^2-10m+4}\right) \lambda_{4}.
\eea
Since $s=2$, \eqref{eqn:ge} becomes
\bea
\label{eqn:ge2}
\lambda_2 & < & \left[ \frac{4m-1}{2(m-1)} \right] \lambda_{3} -\frac{1}{2} \sum_{r=4}^{k} \lambda_r.
\eea
Now, using the fact that $m \geq 4$, relations \eqref{eqn:f2} and \eqref{eqn:ge2} imply
\beann
\left(\frac{5m^2+8m-7}{m-1} \right) \lambda_3
& < &
-\left(m^2+2m+8\right) \lambda_4 -\left(3m^2-10m+4\right) \sum_{r=5}^{k} \lambda_r.
\eeann
For $m$ at least four, the coefficient of $\lambda_3$ in the left-hand-side above is positive, thus the left-hand-side is positive.
However, the right-hand-side is always nonpositive.  Hence, we reached a contradiction.  Therefore, our original claim
that Coffman--Sethi conjecture holds for all instances with $m$ or $k$ at least four is true.
\end{proof}

\begin{theorem}
\label{thm:main}
The LD algorithm has a makespan ratio with a worst-case bound equal to
$\frac{5m-2}{4m-1}$.  Moreover, this bound is achievable for every $m \geq 2$.
\end{theorem}

\begin{proof}
Validity of the bound follows from Theorems \ref{thm:1} and \ref{thm:2}.
Second statement of the theorem is established by the family of instances presented in
Section \ref{sec:intro}.
\end{proof}

{\bf Acknowledgment:}
The work was supported in part by Discovery Grants from NSERC (Natural Sciences and Engineering Research Council of Canada).

\end{document}